\documentclass[12pt]{article}
\usepackage[margin=1in]{geometry}

\usepackage{style/qstyle}

\title{Uncloneable Encryption from Decoupling}

\author[1,2]{Archishna Bhattacharyya \thanks{abhat086@uottawa.ca}}
\author[1,3]{Eric Culf \thanks{eculf@uwaterloo.ca}}
\affil[1]{Perimeter Institute for Theoretical Physics}
\affil[2]{Department of Mathematics and Statistics, University of Ottawa}
\affil[3]{Institute for Quantum Computing and University of Waterloo}
\date{}

\begin{document}

\maketitle

\vspace{0.5cm}

\begin{abstract}
    We show that uncloneable encryption exists with no computational assumptions, with security $\widetilde{O}\parens*{\tfrac{1}{\lambda}}$ in the security parameter $\lambda$.
\end{abstract}

\vspace{2cm}

\tableofcontents

\newpage

\section{Introduction}

Uncloneable cryptography is a paradigm wherein the no-cloning principle of quantum mechanics~\cite{Par70,WZ82,Die82} is used to achieve classically-impossible security guarantees. This underpins much of the groundbreaking work in quantum cryptography, notably quantum key distribution~\cite{BB84} and quantum money~\cite{Wie83}. More recently, Broadbent and Lord~\cite{BL20} introduced the notion of \emph{uncloneable encryption}, which defines a stronger form of security. The goal of an uncloneable encryption scheme is to encode a classical message as a quantum ciphertext in order to guarantee that two non-interacting adversaries cannot both learn the message, even when given the encryption key. This is a security notion that is impossible classically because any classical ciphertext can be copied. Since then, many other uncloneable cryptographic primitives have been studied, \emph{e.g.}, quantum copy-protection~\cite{AK21,ALL+21,CLLZ21,CMP24}, secure software leasing~\cite{ALP21,BJL+21,KNY20arxiv}, quantum functional encryption~\cite{MM24arxiv}, uncloneable decryption~\cite{GZ20eprint,CLLZ21,SW22arxiv}, and uncloneable zero-knowledge proofs~\cite{JK25}.

However, a security proof for uncloneable encryption has been elusive. So far, security has been proven in the quantum random oracle model~\cite{BL20,AKL+22,AKL23}, which is a heuristic model used to provide evidence for cryptographic schemes. A variety of candidate schemes have been proposed, but their security remains unproven in the plain model\footnote{By plain model, we refer to security proven either without any computational assumptions, or with computational assumptions that are well-justified via existing constructions, such as one-way functions.}. For example,~\cite{CHV24arxiv,AB24} presented candidates relying on conjectures about uncloneable forms of indistinguishability obfuscation, and~\cite{BBCNPR24} presented another candidate relying on a conjecture about the value of a monogamy-of-entanglement game. Other work has concentrated on variants of uncloneable encryption, notably with interaction~\cite{BC23arxiv}, with independent decryption keys~\cite{KT22arxiv}, or with quantum keys~\cite{AKY24arxiv}. On the other hand, the possibility that uncloneable encryption is impossible has also been considered, with some work proving no-go theorems on possible schemes~\cite{MST21,AKL+22}. In particular, it is known that the states used to encode the messages must be highly mixed.

In this work, we demonstrate the existence of an uncloneable bit in the information-theoretic model, with \emph{no} computational assumptions. An uncloneable bit is a family of quantum encryptions of classical messages (QECMs) encoding a single bit that can be scaled to have arbitrarily good uncloneable security. Due to~\cite{HKNY24}, an uncloneable bit can be used to construct secure uncloneable encryption schemes for messages of arbitrary length, under standard cryptographic assumptions (see~\cref{sec:uncloneable-encryption} for more details). Hence, the uncloneable bit is a fundamental cryptographic primitive in uncloneable cryptography. Note however that, to the best of our knowledge, it remains an open question whether it is possible to extend the message length of a QECM while preserving information-theoretic uncloneable security. Throughout this work, the property of correctness (\cref{def:qecm}) is implicitly required, as to perfectly decrypt a QECM one requires states that are orthogonal.

The security of an uncloneable bit is defined by means of the following security game, played between an honest referee, Alice, and two cooperating malicious players, Bob and Charlie:
\begin{enumerate}[1.]
    \item Alice samples a random key $k$ and a bit message $x \in\{0,1\}$, and prepares the quantum ciphertext $\sigma^k_x$.
    \item $\sigma^k_x$ passes through an adversarially-chosen pirate channel outputting an entangled state in Bob and Charlie's systems.
    \item Alice informs Bob and Charlie of the key $k$.
    \item Without communicating, Bob and Charlie both try to guess the original message $x$.
    \item The players win if both of their guesses are correct.
\end{enumerate}
This game models a \emph{cloning attack} against Alice's QECM, illustrated in \cref{fig:cloning-attack}. By making a random but coordinated guess of the message bit, Bob and Charlie can always win the game with probability $\frac{1}{2}$. The level of security is quantified by how much better they can do than this in the winning probability. This is defined formally in~\cref{sec:uncloneable-encryption}. For example, if Alice encodes the message in one of the two conjugate-coding bases, Bob and Charlie can win with probability $\cos^2\parens*{\tfrac{\pi}{8}}\approx0.85$~\cite{TFKW13}. We study the Haar-measure encryption of a bit, where, in order to encode a bit, Alice samples a random basis from the Haar measure on the unitary group, and prepares a randomly-chosen state from among either the first or second half of the states in the basis, depending on the value of the bit. See \cref{def:haar-qecm} for the formal definition of this QECM. QECMs based on the Haar measure were first introduced in~\cite{MST21} and have been further studied in~\cite{PRV24}.

The security property that we refer to as \emph{uncloneable security} and show for the uncloneable bit is that, in the limit of large dimension, the success probability of a cloning attack must tend to $\frac{1}{2}$. Specifically, we show that the uncloneable security scales as $O(\tfrac{\log(\log d)}{\log d})$ where $d$ is the dimension, which translates to $\widetilde{O}(\tfrac{1}{\lambda})$ in the security parameter $\lambda$. This differs from the definition first proposed in~\cite{BL20}, which requires the scaling to be negligible in $\lambda$, but coincides with the notion of `weak uncloneable security' introduced in~\cite{BBCNPR24}. The possibility of a negligible bound, which we refer to as \emph{strong uncloneable security}, remains an open problem.

An important property used to study uncloneable encryption is its relationship to monogamy-of-entanglement (MoE) games. \emph{Monogamy} is a property of quantum entanglement~\cite{Terhal04} which asserts that in a tripartite system, if Bob is highly entangled with Alice, then Charlie can only be weakly entangled. One of the ways in which the strength of this property is quantified is via the winning probability of an MoE game. An MoE game, first defined in~\cite{TFKW13}, is a game played between a referee, Alice, and two cooperating players, Bob and Charlie, as follows:
\begin{enumerate}[1.]
    \item Bob and Charlie prepare a state shared amongst themselves and Alice, and then are separated and can no longer communicate.
    \item Alice samples a question $\theta$ and informs Bob and Charlie. 
    \item Alice makes a measurement specified by $\theta$ to get answer $x$.
    \item Bob and Charlie measure their parts of the state and both attempt to guess $x$.
    \item Bob and Charlie win if both their guesses are correct.
\end{enumerate}
The formal definition of an MoE game is given in \cref{sec:moe}. Notably, many QECMs have an associated MoE game where cloning attacks can be mapped to strategies for the game in such a way that the success probability is preserved~\cite{Cul22}. This can be seen as a type of equivalence between prepare-and-measure and entanglement-based schemes. The MoE game analogue of the Haar-measure encryption is illustrated in \cref{fig:moe-game} and defined formally in \cref{def:haar-moe}. Many known bounds on the cloning values of QECMs arise from studying the associated MoE game, for example using the overlap technique~\cite{TFKW13,CV22,CVA22} or the NPA hierarchy~\cite{JMRW16,BBCNPR24}. However, none of the known techniques have been able to give tight enough upper bounds on the MoE game value to demonstrate uncloneable encryption.

We achieve this by means of \textit{decoupling}, originally proposed in the context of error suppression \cite{SW02} — now a well-known paradigm in quantum information theory \cite{Dup10} with notable applications in quantum Shannon theory~\cite{HHWY08}, resource theories \cite{MBDRC17}, and more recently authentication \cite{AM17} and quantum encryption \cite{LM20} in cryptography. We make use of a one-shot refinement due to~\cite{DBWR14}.

\begin{figure}
    \centering
    \begin{tikzpicture}[scale=0.8]
        \begin{scope}[node distance=0.8 and 3.2,thick]
        \filldraw[dotted,draw=red!60!black,fill=black!5!white,rounded corners] (4,-2.8) rectangle (12.6,2.8);
        \node (alice)[shape=rectangle,draw=black,rounded corners] {\Large$A$};
        \node (pirate) [shape=rectangle,draw=red!60!black,rounded corners,right=of alice] {\Large$\Phi$};
        \node (bob) [shape=rectangle,draw=red!60!black,rounded corners,above right=of pirate] {\Large$B$};
        \node (charlie) [shape=rectangle,draw=red!60!black,rounded corners,below right=of pirate] {\Large$C$};
        \draw[-Latex] (alice) -- (pirate) node[above,pos=0.5]{$\sigma^k_x$};
        \draw[-Latex,red!60!black] (pirate) -- (bob);
        \draw[-Latex,red!60!black] (pirate) -- (charlie);
        \draw[red!60!black,decorate,decoration=snake] (bob) -- (charlie);
        \draw[double,Latex-] (alice) -- ++(0,1) node[above]{$k$};
        \draw[double,Latex-] (alice) -- ++(-1.5,0) node[left]{$x$};
        \draw[double,Latex-] (bob) -- ++(0,1) node[above]{$k$};
        \draw[double,Latex-] (charlie) -- ++(0,-1) node[below]{$k$};
        \draw[double,-Latex,red!60!black] (bob) -- ++(1.5,0) node[right]{\color{black}$x_B$};
        \draw[double,-Latex,red!60!black] (charlie) -- ++(1.5,0) node[right]{\color{black}$x_C$};
        \end{scope}
    \end{tikzpicture}
    \caption{Cloning attack against a QECM. The parts depending on the cloning attack are outlined by a dotted box.}
    \label{fig:cloning-attack}
\end{figure}
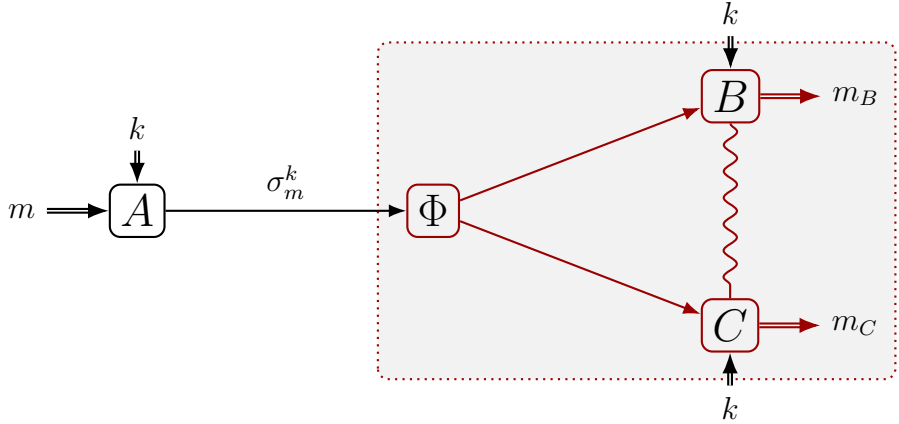

\paragraph{Summary of results} We show that there exists a quantum encryption of classical messages that is correct and uncloneable secure, which is our main result, \cref{thm:main-theorem}. We show this by proving that the quantum value (or the maximum winning probability) of the $d$-dimensional two-outcome Haar measure game is $\frac{1}{2} + O\parens[\big]{\tfrac{\log(\log d)}{\log d}}$ (\cref{thm:D}). The proof holds as a consequence of applying the one-shot decoupling theorem (\cref{thm:decoupling} \cite{DBWR14}) to show that one cannot win significantly better than random guessing in the above type of monogamy-of-entanglement game, where the $O\parens[\big]{\tfrac{\log(\log d)}{\log d}}$ terms comes from the error term in the decoupling theorem (see \cref{sec:decoupling}). We also show how to achieve a correct and uncloneable-secure quantum encryption of classical messages with security $\widetilde{O}\parens*{\tfrac{1}{\lambda}}$ in the security parameter $\lambda$ efficiently by a construction with unitary \textit{t-}designs (\cref{thm: eff-const}). 

\cref{thm:main-theorem} is established as follows. First, we need that the success probability of a cloning attack (\cref{def:cloning-att}) is upper bounded by the winning probability of a two-outcome $d$-dimensional Haar measure game (\cref{lem:moe-bound}). Next, we establish that the winning probability of this game is $\frac{1}{2} + O\parens[\big]{\tfrac{\log(\log d)}{\log d}}$ (\cref{thm:D}) by contradiction. This is where the decoupling theorem plays a role as follows. We bound the value of the conditional min-entropy away from its minimum value in \cref{cor:entropy-bound}, which appears as an exponential bound in the decoupling inequality (\cref{thm:decoupling}). Furthermore, we contradict our initial assumption that there exists a shared state of Alice, Bob and Charlie that corresponds to a winning probability much greater than $\frac{1}{2}$ in the game by showing that its overlap with any other tripartite state where Alice and Charlie decouple must be very small (\cref{thm:inner-product-bound}). Then, by \emph{monogamy of entanglement}, Alice and Bob are not highly entangled. Yet, this overlap being small implies Alice's randomised measurements cause her system to be decoupled from Bob, due to which the probability of Bob correctly guessing Alice's measurement outcome is always low which negates our assumption. This argument is elaborated in \cref{sec:intuition}.

\paragraph{Outline} The rest of this paper is structured as: \cref{sec:intuition} elucidates the intuition behind the proof of the main result; \cref{sec:prelim} contains all necessary background to follow the main result; \cref{sec:result} states and proves \cref{thm:main-theorem}, the main result; \cref{sec:eff-const} presents the accompanying result of an efficient construction, \cref{thm: eff-const}; and finally, \cref{sec:outlook} concludes with an outlook and future directions.

\paragraph{Acknowledgements} We are grateful to Anne Broadbent, David Elkouss, Yao-Ting Lin and anonymous reviewers for helpful discussion and comments on a draft of the manuscript. AB thanks Debbie Leung for teaching her about decoupling. EC thanks everyone with whom he has discussed the uncloneable encryption problem in depth: Pierre Botteron, Srijita Kundu, S\'ebastien Lord, Arthur Mehta, Ion Nechita, Monica Nevins, Cl\'ement Pellegrini, Denis Rochette, Hadi Salmasian, and William Slofstra. Research at Perimeter Institute is supported in part by the Government of Canada through the Department of Innovation, Science, and Economic Development Canada and by the Province of Ontario through the Ministry of Colleges and Universities. EC is supported by a CGS D scholarship from NSERC.

\section{Main Idea} \label{sec:intuition}

The central notion which enables uncloneable encryption is decoupling. In a tripartite pure state $\phi_{ABC}$, we say that systems $A$ and $C$ \textit{decouple} if the reduced state of $\phi_{AC} = \phi_A \otimes \phi_C$ is a product state, in which case $A$ is purified by subsystem $B$. This means that $C$ does not provide any information on $A$, i.e., the outcome of any measurement of $A$ is statistically independent of the outcome of any measurement on $C$. A decoupling inequality provides a necessary and sufficient condition for which the unitary evolution of a system results in decoupling. 

The relevance of decoupling in showing the validity of uncloneable encryption lies in proving that the winning probability of a certain type of monogamy-of-entanglement game is sufficiently bounded. Monogamy of entanglement is the idea that in a tripartite system $ABC$, if $A$ and $B$ are highly entangled, then each of their reduced states with $C$ should be highly separable. In a monogamy-of-entanglement game between three parties Alice, Bob, and Charlie, this implies that both Bob and Charlie cannot simultaneously win with high probability. In other words, Alice's measurement outcome cannot be simultaneously correctly guessed by both Bob and Charlie. 

The precise monogamy-of-entanglement game useful in this setting is the Haar measure game as in \cref{def:haar-moe}. Here, Alice samples her measurement operators from a Haar distribution and measures her first qubit $A_1$, while Bob and Charlie simultaneously try to guess her measurement outcome correctly without interacting with each other. A simple strategy that Bob and Charlie could use is use a coordinated random guess, and this case corresponds to a winning probability of~$\frac{1}{2}$, which is the worst-case scenario\footnote{Note that Bob and Charlie could do worse than $\tfrac{1}{2}$ by making different guesses, but the value of the coordinated random guess can always be attained for any MoE game.}. 

Now, assume that they can do better, which means there exists a state $\psi_{ABC}$ that fares well in this game such that its winning probability is well above $\frac{1}{2}$. Then, the overlap between $\psi_{ABC}$ and any state $\sigma_{ABC}$ where systems $A$ and $C$ decouple is very small. Note that, by symmetry between Bob and Charlie, we could alternately begin with a state where $A$ and $B$ decouple (see further discussion in~\cref{sec:outlook}). By monogamy of entanglement, this would mean that the overlap of $\psi_{ABC}$ and $\sigma_{ABC}$ is also very small when systems $A$ and $B$ are maximally entangled. The decoupling inequality on $\psi_{ABC}$ in this context says that the amount of \textit{independent randomness} (not shared) that can be extracted from $A$ depends on how far away systems $A$ and $B$ are from being maximally entangled, quantified by the overlap with a maximally entangled state. Since we know that this overlap is small for $\psi_{ABC}$, Alice's randomised measurement must result in her system being decoupled from $B$. In the context of decoupling, this measurement is equivalent to Alice applying a Haar-random unitary to her system and then tracing out a subsystem $A_2$ corresponding to all but the first qubit. She then measures the first qubit $A_1$. By virtue of $A_1$ decoupling from $B$, the probability of Bob correctly guessing Alice's measurement outcome is always low. However, our assumption implies the contrary in that Bob must win with probability much higher that $\frac{1}{2}$. This is a contradiction. Hence, there cannot exist a state like $\psi_{ABC}$. 

The relation of the above argument to uncloneable encryption lies in the equivalence between the entanglement-based and the prepare-and-measure picture of the scheme. On one hand, the Haar measure game induces a quantum encryption of classical messages (\cref{def:qecm}) wherein Alice prepares message states in a randomly sampled basis instead of measuring. On the other hand, by the Choi-Jamio\l{}kowski isomorphism, any cloning attack against this quantum encryption of classical messages induces a strategy for the Haar measure game. Since we proved that the winning probability of the Haar measure game is very close to $\frac{1}{2}$ this implies that Alice's encryption scheme is secure against cloning. 

A schematic of this argument is given in \cref{fig:decoupling-implies-uncloneable}.

\begin{figure}
    \centering
    \begin{subfigure}{0.45\textwidth}
        \centering
        \begin{tikzpicture}[scale=0.75]
        \begin{scope}[thick]
            \node (rho) {$\phi_{ABC}$};
            \node (alice)[shape=rectangle,draw=black,rounded corners,left=2 of rho] {\Large$A$};
            \node (bob)[shape=rectangle,draw=red!60!black,rounded corners,above right=1.73 and 1 of rho] {\Large$B$};
            \node (charlie)[shape=rectangle,draw=red!60!black,rounded corners,below right=1.73 and 1 of rho] {\Large$C$};
            \draw[red!60!black,decorate,decoration=snake] (alice) -- (rho) node[above,pos=0.6]{\scalebox{0.65}[0.7]{\color{black}$\mc{U}(\!\ketbra{x}_{A_1}\!\!\!\otimes\!\mds{1}_{A_2}\!)U^\dag$}} (bob) -- (rho) (charlie) -- (rho);
            \draw[double,Latex-] (alice) -- ++(0,1.5) node[left,pos=0.5]{\scalebox{0.9}{$U$}};
            \draw[double,-Latex] (alice) -- ++(-1,0) node[left]{$x$};
            \draw[double,Latex-] (bob) -- ++(0,1) node[above]{$U$};
            \draw[double,-Latex,red!60!black] (bob) -- ++(1,0) node[right]{\color{black}$x_B$};
            \draw[double,Latex-] (charlie) -- ++(0,-1) node[below]{$U$};
            \draw[double,-Latex,red!60!black] (charlie) -- ++(1,0) node[right]{\color{black}$x_C$};
    
            \draw[decorate, decoration={random steps,segment length=1mm,amplitude=0.6mm},inner color=black!30,outer color=white,draw=white!0] (alice) ++(0,2) circle (0.7cm) node{\color{black}$\mc{U}(d)$};
        \end{scope}
        \end{tikzpicture}
        \caption{The Haar measure MoE game}
        \label{fig:moe-game}
    \end{subfigure}
    \begin{subfigure}{0.45\textwidth}
        \centering
        \begin{tikzpicture}[scale=0.75]
        \begin{scope}[thick]
            \node (alice) {\Large$A$};
            \node (bob) [below=2 of alice] {\Large$B$};
            \node (charlie) [below=2 of bob] {\Large$C$};
            \node (alice2) [below right=-0.6 and 3 of alice] {\Large$A_1$};
            \node (bob2) [right=3 of bob] {\Large$B$};
            \node (charlie2) [right=3 of charlie] {\Large$C$};
            \node (u)[shape=rectangle,draw=black,rounded corners,right=1 of alice] {\large$U$};
            \draw[orange!90!black,decorate,decoration=snake] (bob) ++(-2,0) -- (alice);
            \draw[yellow!80!black,decorate,decoration=snake] (bob) ++(-2,0) -- (bob);
            \draw[violet!90!black,decorate,decoration=snake] (bob) ++(-2,0) -- (charlie);
            \fill (bob) ++(-2,0) circle (3pt) node[above left]{$\phi_{ABC}$};
            \draw (bob) -- (bob2) (charlie) -- (charlie2) (alice) -- (u) (alice2) -- ++(-2.4,0);
            \draw[-Latex] (u) -- ++(2.5,1) node[above,pos=0.5,sloped]{$A_2$} node[right]{$\Tr$};
            \draw[green!60!black,decorate,decoration=snake] (bob2) -- (charlie2);
            \draw[blue!60!black,decorate,decoration=snake] (alice2) to[bend left=40] (charlie2);
            
            \draw[white, opacity=0] (charlie) -- ++(0,-1.4);

            \draw[double,Latex-] (u) -- ++(0,1.5);
            \draw[decorate, decoration={random steps,segment length=1mm,amplitude=0.6mm},inner color=black!30,outer color=white,draw=white!0] (u) ++(0,2) circle (0.7cm) node{\color{black}$\mc{U}(d)$};
        \end{scope}
        \end{tikzpicture}
        \caption{The decoupling theorem.}
        \label{fig:decoupling}
    \end{subfigure}

    \begin{subfigure}{\textwidth}
        \centering
        \begin{tikzpicture}
            \draw[line width=0.5cm,red,opacity=0.4] (-0.9,-0.9) -- (0.9,0.9);
            \node (phi) {\Large$\phi_{ABC}$};
            \draw[line width=0.5cm,red!40!white] (phi) circle (1.3cm);
            
            \draw[line width=0.3cm,green,opacity=0.4] (phi) ++(4,0.5) node{\huge\ding{51}} circle (0.75cm);
            \draw[line width=0.3cm,green,opacity=0.4] (phi) ++(-4,0.5) node{\huge\ding{51}} circle (0.75cm);
            \draw[thick,-Latex] (phi) ++(3,0.3) -- ++(-1.4,-0.3);
            \draw[thick,-Latex] (phi) ++(-3,0.3) -- ++(1.4,-0.3);
        \end{tikzpicture}
        \caption{Uncloneable encryption exists}
        \label{fig:exists}
    \end{subfigure}
    \caption{Decoupling implies the existence of uncloneable encryption. It may seem like specific structure on $U$ and $A_2$ is needed for decoupling to be applicable. However, this is not so, as the only information required is that $U$ is Haar random which is true by construction, and that the size of $A_2$ is sufficiently large compared to $A_1$ which naturally arises in decoupling.}
    \label{fig:decoupling-implies-uncloneable}
\end{figure}
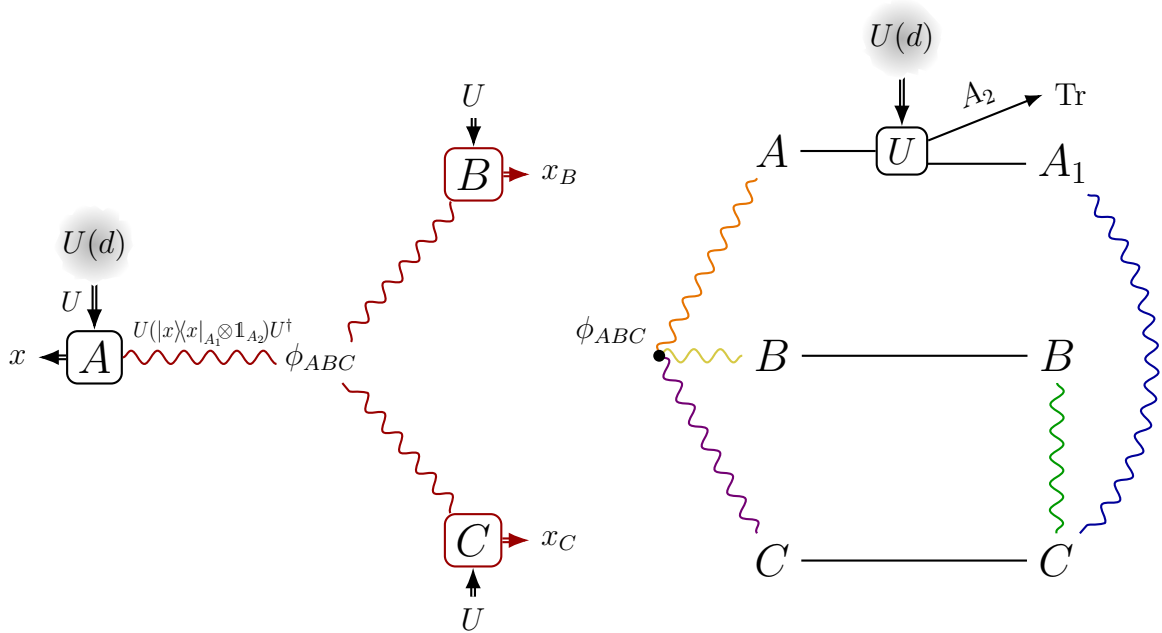
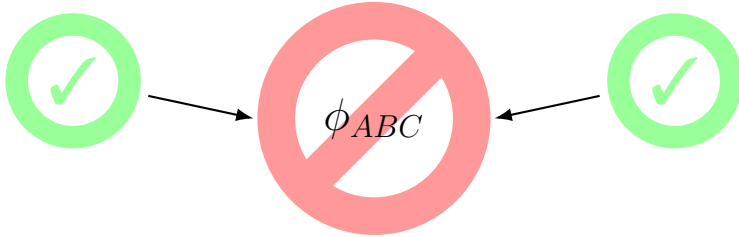

\section{Preliminaries} \label{sec:prelim}
    
\subsection{Notation} \label{sec: notation}

For $n\in\N$, write $[n]=\{1,2,\ldots,n\}$. We write $\log$ for the base-$2$ logarithm. For functions $f,g:\N\rightarrow\R_{\geq 0}$, we say $f(\lambda)=O(g(\lambda))$ if $\lim \limits_{\lambda\rightarrow\infty}\frac{f(\lambda)}{g(\lambda)}<\infty$; and $f(\lambda)=\widetilde{O}(g(\lambda))$ if $f(\lambda)=O(g(\lambda)\log (\lambda)^c)$ for some $c\in\R$. We say a function $f$ is negligible if $\lim \limits_{\lambda\rightarrow\infty}\lambda^n f(\lambda)=0$ for all $n\in\N$.

We denote registers by uppercase Latin letters $A,B,C,\ldots$; and we denote Hilbert spaces by uppercase script letters $\mc{H},\mc{K},\mc{L},\ldots$. We always assume registers are finite sets and Hilbert spaces are finite-dimensional. We denote an independent copy of a register $A$ by $A'$. Given a register $A$, the Hilbert space spanned by $A$ is $\mc{H}_A=\spn\!\!\set*{\ket{a}}{a\in A}\cong\C^{|A|}$. We indicate that an operator or vector is on register $A$ with a subscript $A$, omitting when clear from context. Given two registers $A$ and $B$, we write $AB$ for their cartesian product, and treat the isomorphism $\mc{H}_{AB}\cong\mc{H}_A\otimes\mc{H}_B$ implicitly. Given finite-dimensional Hilbert spaces $\mc{H}$ and $\mc{K}$, we write $B(\mc{H},\mc{K})$ for the set of all linear operators $\mc{H}\rightarrow\mc{K}$, $B(\mc{H})=B(\mc{H},\mc{H})$, $\mc{U}(\mc{H})\subseteq B(\mc{H})$ for the subset of unitary operators, and $D(\mc{H})\subseteq B(\mc{H})$ for the subset of density operators. Write $\mc{U}(d)=\mc{U}(\C^{d})$. We write $\Tr$ for the trace on $B(\mc{H})$. On $B(\mc{H}_{AB})$, we write the partial trace $\Tr_{B}=\id\otimes\Tr$. For $\rho_{AB}\in B(\mc{H}_{AB})$, write $\rho_A=\Tr_B(\rho_{AB})$. We denote the $1$-norm by $\norm{\cdot}_1$ and the trace norm by $\norm{\cdot}_{\Tr}=\frac{1}{2}\norm{\cdot}_1$.

We denote the canonical maximally-entangled state $\ket{\phi^+}_{AA'}=\frac{1}{\sqrt{|A|}}\sum_{a\in A}\ket{a}\otimes\ket{a}\in\mc{H}_{AA'}$. We write the maximally-mixed state on a register $A$ as $\omega_A=\frac{1}{|A|}\sum_{a\in A}\ketbra{a}\in D(\mc{H}_A)$. 

A positive-operator-valued measurement (POVM) is a finite set of positive operators $\{P_i\}_{i\in I}$ such that $\sum_iP_i=\mds{1}$, and a projection-valued measurement (PVM) is a POVM where all the elements are projectors. A quantum channel is a completely positive trace-preserving (CPTP) map $\Phi:B(\mc{H})\rightarrow B(\mc{K})$. We denote the Choi-Jamio\l{}kowski isomorphism $J:B(B(\mc{H}_A),B(\mc{H}_B))\rightarrow B(\mc{H}_{AB})$, $J(\Phi)=(\id\otimes\Phi)(\ketbra{\phi^+}_{AA'})$. Note that if is $\Phi$ is a quantum channel, $J(\Phi)\in D(\mc{H}_{AB})$, called the Choi state.

For a complex-valued random variable $X$, we write its expectation as $\expec X=\expec_XX$, and its variance as $\varsigma_X^2=\expec |X|^2-\abs*{\expec X}^2$. We make use of the Haar measure on the unitary group, which is the unique invariant Borel probability measure on $\mc{U}(\mc{H})$, for $\mc{H}$ a finite-dimensional Hilbert space. We denote it $\mu_{\mc{U}(\mc{H})}$. Given a function $f$ with domain $\mc{U}(\mc{H})$, we interchangeably write $\expec_U f(U)=\int f(U)dU$ for the expectation with respect to the Haar measure.



\subsection{Representation theory}\label{sec:representations}

Let $G$ be a finite or a compact topological group. A \emph{unitary representation} of $G$ on a finite-dimensional Hilbert space $\mc{H}$ is a group homomorphism $\pi:G\rightarrow \mc{U}(\mc{H})$. If $G$ is a topological group, we will also require that $\pi$ be continuous. An \emph{intertwiner} from a representation $\pi:G\rightarrow \mc{U}(\mc{H})$ to a representation $\chi:G\rightarrow \mc{U}(\mc{K})$ is an operator $T\in B(\mc{H},\mc{K})$ such that $\chi(g)T=T\pi(g)$ for all $g\in G$. A natural way to construct intertwiners is by means of the Haar measure on $G$, $\mu_G$. In fact, if $T\in B(\mc{H},\mc{K})$,
$$\int\chi(g)T\pi(g)^\dag d\mu_G(g)$$
is always an intertwiner. We say two representations are \emph{equivalent} if there is an invertible intertwiner between them and write $\pi\simeq\chi$. An \emph{irreducible representation} is a representation whose action on $\mc{H}$ leaves no subspace but $\mc{H}$ and $0$ invariant. By Schur's lemma, the intertwiners between inequivalent irreducible representations are $0$ and the intertwiners from an irreducible representation to itself are multiples of identity. For finite groups (Maschke's theorem) or compact topological groups (Peter-Weyl theorem), every representation decomposes as a direct sum of irreducibles, \emph{i.e.} given a representation $\pi:G\rightarrow \mc{U}(\mc{H})$ there exists an equivalence $\mc{H}\rightarrow\bigoplus_i\mc{H}_i\otimes\mc{K}_i$ such that $\pi\simeq\bigoplus_i\pi_i\otimes \mds{1}$, where the $\pi_i:G\rightarrow \mc{U}(\mc{H}_i)$ are inequivalent irreducible representations. The intertwiners from $\pi$ to itself then have the form $T=\bigoplus_i\mds{1}\otimes T_i$ for some $T_i\in B(\mc{K}_i)$.

We work with representations of the unitary group on a finite-dimensional Hilbert space $\mc{H}$. Fix a basis $\set*{\ket{i}}{i=1,\ldots,d}$ of $\mc{H}$. The \emph{trivial representation} is the mapping $\mc{U}(\mc{H})\rightarrow S^1$, $U\mapsto 1$; the \emph{fundamental representation} is the identity mapping $\mc{U}(\mc{H})\rightarrow \mc{U}(\mc{H})$; and the \emph{contragredient representation} is the mapping $\mc{U}(\mc{H})\rightarrow \mc{U}(\mc{H})$, $U\mapsto\bar{U}$, where the complex conjugate is with respect to the fixed basis. These are all irreducible representations, and inequivalent for $d>2$. 

\begin{lemma} [\cite{Mel24}] \label{lem: uni-irrep}
    Consider the representation $\pi:\mathcal{U}(\mathcal{H})\rightarrow \mathcal{U}(\mathcal{H}\otimes\mathcal{H})$, $U\mapsto U\otimes\bar{U}$. Then, any intertwiner $T$ of $\pi$ can be expressed as $T=\alpha|\phi^+\rangle\!\langle \phi^+|+\beta\Pi$,
    where $\Pi={I}-|\phi^+\rangle\!\langle \phi^+|$ is the orthogonal projector onto $\mathcal{K}=|\phi^+\rangle^\perp$, where $|\phi^+\rangle\in\mathcal{H}\otimes\mathcal{H}$ is the maximally entangled state.\footnote{This is equivalent to the decomposition as the direct sum of two irreducible representations: the trivial representation on the subspace $\mathrm{span}\{|\phi^+\rangle\}$ and an irreducible representation on $\mathcal{K}=|\phi^+\rangle^\perp$.} In particular, by orthogonality of the projectors,
    \begin{align}
        \int(U\otimes\bar{U})T(U\otimes\bar{U})^\dag dU=\frac{\mathrm{Tr}(\Pi T)}{d^2-1}\Pi+\langle{\phi^+}|{T}|{\phi^+}\rangle|\phi^+\rangle\!\langle \phi^+|.
    \end{align}
\end{lemma}
For more details on the representation theory of the unitary group, see for example~\cite{Mel24}.

\subsection{Unitary \textit{t}-designs}

Unitary $t$-designs give a way to replace the Haar measure over the unitary group by a finitely-supported measure.

\begin{definition}
    Let $\mc{H}$ be a finite-dimensional Hilbert space and $t\in\N$. A \emph{unitary $t$-design} on $\mc{H}$ is a finite subset $\mc{V} \subseteq \mc{U}(\mc{H})$ such that
    \begin{align}
        \frac{1}{|\mc{V}|}\sum_{U\in \mc{V}}U^{\otimes t}\otimes\bar{U}^{\otimes t}=\int U^{\otimes t}\otimes\bar{U}^{\otimes t} dU.
    \end{align}
\end{definition}

For any function $p:\mc{U}(\mc{H})\rightarrow\C$ such that $p(U)$ is a degree-$t$ polynomial in the matrix elements of $U$ and $\bar{U}$, the $t$-design property implies that $\frac{1}{|\mc{V}|}\sum_{U\in \mc{V}}p(U)=\expec_Up(U)$.

In~\cite{DLT02}, it was shown that the Clifford group induces a $2$-design. Further,~\cite{DCEL09,CLLW16} show that $2$-designs can be efficiently implemented.

\subsection{Entropies}

We recall that the von Neumann entropy of a state $\rho \in D(\mc{H})$ is defined as $H(\rho) = -\Tr(\rho \log \rho)$. The conditional von Neumann entropy of A given B for a bipartite state $\rho_{AB} \in D(\mc{H}_A \otimes \mc{H}_B)$ is defined as $H(A \vert B)_{\rho} = H(AB)_{\rho} - H(B)_{\rho}$. In this work, we shall be using the conditional min-entropy denoted by $H_{\min} (A \vert B)_{\rho}$ defined below.

\begin{definition} \label{def: min-ent}
    Let $\rho_{AB} \in D(\mc{H}_{AB})$. The conditional min-entropy of A given B is defined as
    \begin{equation}
        H_{\min} (A \vert B)_{\rho} = \sup \limits_{\sigma_B \in D(\mc{H}_B)} \sup \{ \lambda \in \mbb{R} : 2^{-\lambda} \cdot \mathds{1}_A \otimes \sigma_B \geq \rho_{AB} \}.
    \end{equation}
\end{definition}

These entropies are related as $H_{\min} (A \vert B)_{\rho} \leq H (A \vert B)_{\rho}$ for all $\rho \in D(\mc{H}_{AB})$ \cite[Lemma 2]{TCR09}. Operationally, the conditional min-entropy quantifies how close the state $\rho_{AB}$ can be brought to a maximally entangled state on the bipartite system $AB$ using only local quantum operations on system $B$.

There is an alternate operational interpretation of the min-entropy that is crucial to this work. The min-entropy is the \textit{maximum achievable ebit fraction} \cite{KRS09}, where an ebit\footnote{An "ebit" is the maximally entangled two-qubit state $\frac{1}{\sqrt{2}}(\ket{00} + \ket{11})$, but here an ebit refers to any general maximally entangled state $\ket{\phi^+}_{AA'}$. This is justified as the quantity $|A| \max \limits_{\mc{E}} F \left((\text{id}_A \otimes \mc{E})(\rho_{AB}), \ketbra{\phi^+} \right)^2$ takes the same value independent of the choice of maximally entangled state $\ket{\phi^+}_{AA'}$ as originally formulated \cite{KRS09}.} is the maximally entangled state defined by 

\begin{equation} \label{eq: max-ent-state}
    \ket{\phi^+}_{AA'}=\frac{1}{\sqrt{|A|}}\sum_{a\in A}\ket{a}\otimes\ket{a}\in\mc{H}_{AA'}.
\end{equation} 

The maximum overlap of a state $\rho_{AB}$ with an ebit that can be achieved by local quantum operations (CPTP maps) $\mc{E}: B(\mc{H}_B) \to B(\mc{H}_{A'}), ~\mc{H}_{A'} \cong \mc{H}_{A}$ on subsystem B is the quantity $|A| \max \limits_{\mc{E}} F \left((\text{id}_A \otimes \mc{E})(\rho_{AB}), \ketbra{\phi^+} \right)^2$, interpreted as the amount of quantum correlation between $A$ and $B$. Here, $F$ is the fidelity between quantum states $\rho$ and $\sigma$ denoted by $F (\rho, \sigma)$ defined as $$F (\rho, \sigma) \coloneqq \left\Vert \sqrt{\rho} \sqrt{\sigma} \right\Vert_1.$$

The conditional min-entropy $H_{\min}(A \vert B)_{\rho} = - \log |A| \max \limits_{\mc{E}} F \left((\text{id}_A \otimes \mc{E})(\rho_{AB}), \ketbra{\phi^+} \right)^2$ is then interpreted as the negative logarithm of the quantum correlation between $A$ and $B$ or the maximum achievable ebit fraction. This idea is formalised in the following theorem by K\"onig, Renner and Schaffner.

\begin{lemma} \cite[Theorem 2]{KRS09} \label{lem: op-int-min-ent}
    The min-entropy of a state $\rho_{AB} \in D(\mc{H}_A \otimes \mc{H}_B$ can be expressed as 
    \begin{equation}
        H_{\min}(A \vert B)_{\rho} = - \log |A| \max \limits_{\mc{E}} F \left((\text{id}_A \otimes \mc{E})(\rho_{AB}), \ketbra{\phi^+} \right)^2,
    \end{equation}
    with maximum taken over all quantum channels $\mc{E}: B(\mc{H}_B) \to B(\mc{H}_{A'}), ~\mc{H}_{A'} \cong \mc{H}_{A}$ and $\ket{\phi^+}_{AA'}$ defined by \eqref{eq: max-ent-state}.
\end{lemma}

\subsection{Uncloneable encryption}\label{sec:uncloneable-encryption}

\begin{definition} \label{def:qecm}
    \begin{itemize}
    \item A \emph{quantum encryption of classical messages (QECM)} is given by a tuple $\ttt{Q}=(K,X,A,\mu,\{\sigma^k_x\}_{k\in K,x\in X})$, where
    \begin{itemize}
        \item $K$ is a set, representing the encryption keys;
        \item $X$ is a finite set, representing the messages;
        \item $A$ is a register, representing the system holding the encrypted messages;
        \item $\mu$ is a probability measure on $K$, representing the key distribution;
        \item $\sigma^k_x\in D(\mc{H}_A)$ is a quantum state, representing the encryption of message $x$ with key $k$.
    \end{itemize}
    
    \item We say a QECM is \emph{$\eta$-correct} if there exists a family of CPTP maps $\Phi^k:B(\mc{H}_A)\rightarrow B(\mc{H}_M)$, called decryption maps, such that for all $k\in K$ and $x\in M$,
    \begin{align*}
        \braket{x}{\Phi^k(\sigma^k_x)}{x}\geq\eta.
    \end{align*}
    
    \item We say that the QECM is \emph{correct} if it is $1$-correct.
    
    \item We say a family of QECMs $\{\ttt{Q}_\lambda\}_{\lambda\in\N}$ is an \emph{efficient QECM} if key sampling, encrypted message preparation, and decryption can be implemented in polynomial time in $\lambda$.
    \end{itemize}
\end{definition}

Note that correctness is equivalent to the orthogonality condition $\Tr(\sigma^k_{x}\sigma^k_{x'})=0$ for $k\in K$ and $x\neq x'\in X$.

\begin{definition}\label{def:haar-qecm}
    Let $d\in\N$ be even, let $A_1=\{0,1\}$, and $A_2=[d/2]$. Set $A=A_1A_2$. Let $\sigma_0=\frac{2}{d}\ketbra{0}\otimes \mds{1}\in D(\mc{H}_A)$ and $\sigma_1=\frac{2}{d}\ketbra{1}\otimes \mds{1}\in D(\mc{H}_A)$. The \emph{$d$-dimensional Haar-measure encryption of a bit} is the QECM $\ttt{Q}_{d,2}=(\mc{U}(\mc{H}_A),\{0,1\},A,\mu_{\mc{U}(\mc{H}_A)},\{U\sigma_xU^\dag\}_{U\in\mc{U}(\mc{H}_A),x\in\{0,1\}})$.
\end{definition}

\begin{definition} \label{def:cloning-att}
    A \emph{cloning attack} against a QECM $\ttt{Q}=(K,X,A,\mu,\{\sigma^k_x\}_{k\in K,x\in X})$ is a tuple $\ttt{A}=(B,C,\{B^k_x\}_{k\in K,x\in X},\{C^k_x\}_{k\in K,x\in X},\Phi)$, where
    \begin{itemize}
        \item $B$ and $C$ are registers, representing Bob and Charlie's systems, respectively;
        \item $\{B^k_x\}_{x\in X}\subseteq B(\mc{H}_B)$ and $\{C^k_x\}_{x\in X}\subseteq B(\mc{H}_C)$ are POVMs, representing Bob and Charlie's measurements given key $k$, respectively;
        \item $\Phi:B(\mc{H}_A)\rightarrow B(\mc{H}_{BC})$ is a CPTP map, representing the cloning channel.
    \end{itemize}
    The \emph{success probability} of $\ttt{A}$ against $\ttt{Q}$ is
    \begin{align}
        \mfk{c}(\ttt{Q},\ttt{A})=\int\frac{1}{|X|}\sum_{x\in X}\Tr\squ*{(B^k_x\otimes C^k_x)\Phi(\sigma^k_x)}d\mu(k).
    \end{align}
    The \emph{cloning value} of $\ttt{Q}$ is $\mfk{c}(\ttt{Q})=\sup_{\ttt{A}}\mfk{c}(\ttt{Q},\ttt{A})$, where the supremum is over all cloning attacks. We say a QECM is \emph{$\delta$-uncloneable secure} if $\mfk{c}(\ttt{Q})\leq\frac{1}{|X|}+\delta$.

    For a function $f:\N\rightarrow[0,1]$, we say a family of QECMs $\{\ttt{Q}_\lambda\}$ is \emph{$f$-uncloneable secure} if $\ttt{Q}_\lambda$ is $f(\lambda)$-uncloneable secure for all $\lambda$. We additionally say $\{\ttt{Q}_\lambda\}$ is \emph{uncloneable secure} if $\lim\limits_{\lambda\rightarrow\infty}f(\lambda)=0$; and $\{\ttt{Q}_\lambda\}$ is \emph{strongly uncloneable secure} if $f$ is a negligible function.
\end{definition}

\begin{definition}\label{def:cloning-distinguishing-att}
    A \emph{cloning-distinguishing attack} against a QECM $\ttt{Q}=(K,X,A,\mu,\{\sigma^k_x\}_{k\in K,x\in X})$ is a tuple $\ttt{A}=(\{x_0,x_1\},B,C,\{B^k_b\}_{k\in K,b\in\{0,1\}},\{C^k_b\}_{k\in K,b\in\{0,1\}},\Phi)$, where
    \begin{itemize}
        \item $x_0\neq x_1\in X$ are distinct messages, representing the two messages to be distinguished;
        \item $B$ and $C$ are registers, representing Bob and Charlie's systems, respectively;
        \item $\{B^k_b\}_{b\in\{0,1\}}\subseteq B(\mc{H}_B)$ and $\{C^k_b\}_{b\in\{0,1\}}\subseteq B(\mc{H}_C)$ are POVMs, representing Bob and Charlie's measurements given key $k$, respectively;
        \item $\Phi:B(\mc{H}_A)\rightarrow B(\mc{H}_{BC})$ is a CPTP map, representing the cloning channel.
    \end{itemize}
    The \emph{success probability} of $\ttt{A}$ against $\ttt{Q}$ is
    \begin{align}
        \mfk{cd}(\ttt{Q},\ttt{A})=\int\frac{1}{2}\sum_{b\in\{0,1\}}\Tr\squ*{(B^k_b\otimes C^k_b)\Phi(\sigma^k_{x_b})}d\mu(k).
    \end{align}
    The \emph{cloning-distinguishing value} of $\ttt{Q}$ is $\mfk{cd}(\ttt{Q})=\sup_{\ttt{A}}\mfk{cd}(\ttt{Q},\ttt{A})$, where the supremum is over all cloning-distinguishing attacks. We say a QECM is \emph{$\delta$-uncloneable-indistinguishable secure} if~$\mfk{cd}(\ttt{Q})\leq\frac{1}{2}+\delta$.

    For a function $f:\N\rightarrow[0,1]$, we say a family of QECMs $\{\ttt{Q}_\lambda\}$ is \emph{$f$-uncloneable-in\-dis\-ting\-ui\-sha\-ble secure} if $\ttt{Q}_\lambda$ is $f(\lambda)$-uncloneable-indistiguishable secure for all $\lambda$. We additionally say~$\{\ttt{Q}_\lambda\}$ is \emph{uncloneable-indistinguishable secure} if $\lim\limits_{\lambda\rightarrow\infty}f(\lambda)=0$; and $\{\ttt{Q}_\lambda\}$ is \emph{strongly uncloneable-indistinguishable secure} if $f$ is a negligible function.
\end{definition}

Note that if $X=\{0,1\}$, then the notions of uncloneable security and uncloneable-in\-dis\-ting\-ui\-sha\-ble security are equivalent. In general, uncloneable-indistinguishable security implies uncloneable security~\cite{BL20}. Also due to \cite{BL20}, uncloneable-indistinguishable security implies the indistinguishable security, a standard cryptographic notion. Further, due to~\cite{HKNY24}, an uncloneable-indistinguishable secure QECM with a one-bit message can be used to construct an uncloneable-indistinguishable secure QECM with arbitrary message size, under the assumption of quantum polynomial-time adversaries, and a primitive called decomposable quantum randomised encoding, which follows from the existence of one-way functions~\cite{BY22}. Hence, we concentrate on uncloneable security for QECMs with one-bit messages.

\subsection{Monogamy-of-entanglement games}\label{sec:moe}

\begin{definition}
    A \emph{monogamy-of-entanglement (MoE) game} is a tuple $\ttt{G}=(\Theta,X,A,\mu,\{A^\theta_x\}_{\theta\in\Theta,x\in X})$, where
    \begin{itemize}
        \item $\Theta$ is a set, representing the questions;
        \item $X$ is a finite set, representing the answers;
        \item $A$ is a register, representing Alice's system;
        \item $\mu$ is a probability measure on $\Theta$, representing the question distribution.
        \item $\{A^\theta_x\}_{x\in X}\subseteq B(\mc{H}_A)$ is a POVM, representing Alice's measurements given question $\theta$.
    \end{itemize}
    A \emph{strategy} for an MoE game $\ttt{G}$ is a tuple $\ttt{S}=(B,C,\{B^\theta_x\}_{\theta\in\Theta,x\in X},\{C^\theta_x\}_{\theta\in \Theta,x\in X},\rho_{ABC})$, where
    \begin{itemize}
        \item $B$ and $C$ are registers, representing Bob and Charlie's systems, respectively;
        \item $\{B^\theta_x\}_{x\in X}\subseteq B(\mc{H}_B)$ and $\{C^\theta_x\}_{x\in X}\subseteq B(\mc{H}_C)$ are POVMs, representing Bob and Charlie's measurements given question $\theta$, respectively;
        \item $\rho_{ABC}\in D(\mc{H}_{ABC})$ is the shared quantum state.
    \end{itemize}
    The \emph{winning probability} of $\ttt{S}$ at $\ttt{G}$ is
    \begin{align}
        \mfk{w}(\ttt{G},\ttt{S})=\int\sum_{x\in X}\Tr\squ*{(A^\theta_x\otimes B^\theta_x\otimes C^\theta_x)\rho_{ABC}}d\mu(\theta).
    \end{align}
    The \emph{quantum value} of $\ttt{G}$ is $\mfk{w}(\ttt{G})=\sup_{\ttt{S}}\mfk{w}(\ttt{G},\ttt{S})$, where the supremum is over all strategies.
\end{definition}

\begin{definition}\label{def:haar-moe}
    Let $d\in\N$ be even, let $A_1=\{0,1\}$, and $A_2=[d/2]$. Set $A=A_1A_2$. Let $\Pi_0=\ketbra{0}\otimes \mds{1}$ and $\Pi_1=\ketbra{1}\otimes \mds{1}$. The \emph{$d$-dimensional $2$-answer Haar-measure game} is the MoE game $\ttt{G}_{d,2}=(\mc{U}(\mc{H}_A),\{0,1\},A,\mu_{\mc{U}(\mc{H}_A)},\{U\Pi_xU^\dag\}_{U\in\mc{U}(\mc{H}_A),x\in\{0,1\}})$.
\end{definition}

\begin{lemma}\label{lem:moe-bound}
    Let $d\in\N$ be even. Then, $\mfk{c}(\ttt{Q}_{d,2})\leq\mfk{w}(\ttt{G}_{d,2})$.
\end{lemma}

This result generalises to a very wide class of QECM schemes, see for example~\cite[Proposition 5.14]{Cul22}.

\begin{definition}
    Given a strategy $\ttt{S}$ for an MoE game $\ttt{G}$ with single bit answers $X=\{0,1\}$, the \emph{observable form} of the winning probability is
    \begin{align*}
        \mfk{w}(\ttt{G},\ttt{S})=\frac{1}{4}+\frac{1}{4}\int\braket{\psi}{A_\theta\otimes B_\theta\otimes I+A_\theta\otimes I\otimes C_\theta+I\otimes B_\theta\otimes C_\theta}{\psi}d\mu(\theta),
    \end{align*}
    where $A_\theta=A^\theta_0-A^\theta_1$, $B_\theta=B^\theta_0-B^\theta_1$, and $C_\theta=C^\theta_0-C^\theta_1$ are the observables corresponding to the measurements used in the game.
\end{definition}

When working with the observable form, it is useful to consider the \emph{game polynomial}
$$O=\frac{1}{4}\int A_\theta\otimes B_\theta\otimes I+A_\theta\otimes I\otimes C_\theta+I\otimes B_\theta\otimes C_\theta d\mu(\theta).$$
We can always work with symmetric strategies where $B_\theta=C_\theta$ for all $\theta\in\Theta$ and the shared state is preserved under exchange of systems $B$ and $C$.

\subsection{Decoupling theorem}\label{sec:decoupling}

We recall the version of the one-shot decoupling inequality from~\cite{DBWR14} which holds in full generality in terms of the conditional smooth min-entropy $H^{\epsilon}_{\min}(A \vert B)_{\rho}$. We work in the setting where $\epsilon = 0$ and the RHS of the inequality is bounded in terms of the conditional min-entropy $H_{\min} (A \vert B)_{\rho}$. Note that in our restatement of the one-shot decoupling theorem, an additional $-1$ appears in the exponent on the RHS of \cref{eq: decoupling}, whereas in the original formulation \cite{DBWR14}, it doesn't. This is because \cite[Section 2.1]{DBWR14} uses a definition of the trace norm where they omit the factor of $\frac{1}{2}$ which normalises the trace norm, whereas we retain the factor of $\frac{1}{2}$ in our definition of the normalised trace norm (see \cref{sec: notation}).

\begin{theorem}[\cite{DBWR14}]\label{thm:decoupling}
    Let $\rho_{AE}\in D(\mc{H}_{AE})$ be a quantum state, and $\Phi:B(\mc{H}_A)\rightarrow B(\mc{H}_B)$ be a quantum channel. Then,
    \begin{align} \label{eq: decoupling}
    \int\norm*{(\Phi\otimes\id)\squ*{(U\otimes \mds{1})\rho_{AE}(U\otimes \mds{1})^\dag}-\tau_B\otimes\rho_E}_{\Tr}dU\leq2^{-\frac{1}{2}H_{\min}(A|E)_\rho-\frac{1}{2}H_{\min}(A|B)_{\tau}-1},
    \end{align}
    where $\tau_{AB}=J(\Phi)$ is the Choi state of $\Phi$.
\end{theorem}

The channel we will be considering is the partial trace. That is, we decompose the register $A=A_1A_2$ and let $\Phi=\Tr_{A_2}:B(\mc{H}_{A_1A_2})\rightarrow B(\mc{H}_{A_1})$. Then, $\tau_{AA_1'}=\omega_{A_2}\otimes\ketbra{\phi^+}_{A_1A_1'}$, so $H_{\min}(A|A_1')_\tau=\log|A_2|-\log|A_1|$.

\section{Main Result} \label{sec:result}

First, we state the main result of this work, which is the following theorem. We relegate the proof to the end of this section.

\begin{theorem}\label{thm:main-theorem}
    The family of QECMs $\{\ttt{Q}_{2\lambda,2}\}_{\lambda\in\N}$ is correct (as in~\cref{def:qecm}) and uncloneable secure (as in~\cref{def:cloning-att}).
\end{theorem}

Note that the standard notion used for the security of a QECM is uncloneable-indistinguishable security (as in~\cref{def:cloning-distinguishing-att}). However, as noted in~\cref{sec:uncloneable-encryption}, the security guarantees coincide for QECMs encoding a one-bit message, and hence we work with the conceptually simpler notion of uncloneable security.

In this section, we fix $d$ and assume that we are working with a strategy for~$\ttt{G}_{d,2}$, which we denote~$\ttt{S}=(B,C,\{B^U_x\},\{C^U_x\},\rho_{ABC})$.

In the following lemma, we show that the expectation of a certain linear combination of random variables is close to the average value of each random variable with error bounded in terms of the variance.

\begin{lemma}\label{lem:exp-bound}
    Let $T_1,\ldots,T_N$ be complex random variables such that $\expec T_i=\mu$ for all $i$; and let $S_1,\ldots,S_N$ be complex random variables such that $|S_i|\leq 1$ and $\sum_iS_i=1$. Then,
    \begin{align}
        \abs[\Big]{\expec\sum_iT_iS_i-\mu}\leq N\varsigma,
    \end{align}
    where $\varsigma^2=\max_i\varsigma_{T_i}^2$, the maximal variance of the variables $T_i$.
\end{lemma}

\begin{proof}
    Using the Cauchy-Schwarz inequality,
    \begin{align*}
        \abs[\Big]{\expec\sum_i T_i S_i-\mu}&=\abs[\Big]{\expec\sum_i(T_i-\mu)S_i}\\
        &\leq\sqrt{\expec\sum_i|T_i-\mu|^2}\sqrt{\expec\sum_i|S_i|^2}\\
        &\leq\sqrt{\sum_i\varsigma_{T_i}^2}\sqrt{N}\\
        &\leq N\varsigma.\qedhere
    \end{align*}
\end{proof}

The next lemma bounds the mean and variance of random variables generated by Haar-random projections. The bound on the standard deviation scales with the inverse of the dimension.

\begin{lemma}\label{lem:variance}
    Let $d$ be even, and let $A\in B(\mathbb{C}^d)$ be a traceless hermitian unitary. Define the random variable $T=\frac{1}{d}\Tr\parens*{U_1AU_1^\dag U_2AU_2^\dag\cdots U_nAU_n^\dag}$, where the unitaries $U_1,\ldots,U_n$ are i.i.d. random samples from the Haar measure in dimension $d$. Then, the expectations
    \begin{align}
    \begin{split}
        &\expec T=0,\\
        &\expec |T|^2=\parens*{-\frac{1}{d^2-1}}^n\frac{d^2-1}{d^2}+\frac{1}{d^2}\leq\frac{2}{d^2}
    \end{split}
    \end{align}
\end{lemma}

Therefore $\varsigma_T\leq\frac{\sqrt{2}}{d}$.

\begin{proof}
    Using the fact that $\expec U_iAU_i^\dag=0$, we have that
    \begin{align*}
        \expec T=\frac{1}{d}\Tr\parens[\Big]{\prod_i\expec U_iAU_i^\dag}=\frac{1}{2^n}\frac{1}{d}\Tr(0)=0.
    \end{align*}
    Next, note that
    \begin{align*}
        \expec |T|^2&=\expec\frac{1}{d}\Tr\parens[\Big]{\prod_i U_iAU_i^\dag}\frac{1}{d}\Tr\parens[\Big]{\prod_i A\Pi_{x_i}\bar{U}_i^\dag}\\
        &=\expec\frac{1}{d^2}\Tr\parens[\Big]{\prod_i U_iAU_i^\dag\otimes\prod_i \bar{U}_iA\bar{U}_i^\dag}\\
        &=\frac{1}{d^2}\Tr\parens[\Big]{\prod_i\expec(U_i\otimes \bar{U}_i)(A\otimes A)(U_i\otimes \bar{U}_i)^\dag}.
    \end{align*}
    By \cref{lem: uni-irrep}, we have that $$\int(U\otimes\bar{U})X(U\otimes\bar{U})^\dag dU=\frac{\Tr(\Pi X)}{d^2-1}\Pi+\braket{\phi^+}{X}{\phi^+}\ketbra{\phi^+}$$ for any $X\in B(\mc{H}_{AA'})$, where $\Pi=\mds{1}-\ketbra{\phi^+}$. Since the values $\braket{\phi^+}{A\otimes A}{\phi^+}=1$ and  $\Tr(A\otimes A)=0$, we have that
    \begin{align*}
        \expec(U_i\otimes \bar{U}_i)(A\otimes A)(U_i\otimes \bar{U}_i)^\dag&=-\frac{1}{d^2-1}\Pi+\ketbra{\phi^+}.
    \end{align*}
    As such,
    \begin{align*}
        \expec |T|^2&=\frac{1}{d^2}\Tr\parens*{\parens*{\parens*{-\frac{1}{d^2-1}}\Pi+\ketbra{\phi^+}}^n}\\
        &=\frac{1}{d^2}\Tr\parens*{\parens*{-\frac{1}{d^2-1}}^n\Pi+\ketbra{\phi^+}}\\
        &=\parens*{-\frac{1}{d^2-1}}^n\frac{d^2-1}{d^2}+\frac{1}{d^2}.
    \end{align*}
    For the upper bound, note that 
    \begin{align}
       \parens*{-\frac{1}{d^2-1}}^n\frac{d^2-1}{d^2}&=\frac{(-1)^n}{d^2(d^2-1)^{n-1}}\leq\frac{1}{d^2}.\qedhere
    \end{align}
\end{proof}

Finally, in the following technical steps, we show that a maximally-entangled state between Alice and Bob behaves like an eigenvector with eigenvalue $\frac{1}{4}$.

\begin{lemma}\label{lem:near-eigenstate}
    Let $O_{AC}=\frac{1}{4}\int UAU^\dag\otimes I\otimes B_UdU$, where $A$ and $B_U$ are order-two unitaries. Let $\sigma_{ABC}=\ketbra{\phi}_{ABC}$ be such that $\sigma_{AC}=\omega_A\otimes\sigma_C$, where $\omega_A$ is the maximally mixed state on $A$. Then, $\braket{\phi}{O_{AC}^n}{\phi}\leq\frac{\sqrt{2}}{4^nd}$.
\end{lemma}

\begin{proof}
    Write $A_U=UAU^\dag$. Then,
    \begin{align*}
        \braket{\phi}{O_{AC}^n}{\phi}=\frac{1}{4^n}\int\cdots\int\braket{\phi}{A_{U_1}\cdots A_{U_n}\otimes I\otimes B_{U_1}\cdots B_{U_n}
        }{\phi}dU_1\cdots dU_n.
    \end{align*}
    Now, since $\Tr_{AB}\ketbra{\phi}=\sigma_C$, this gives
    \begin{align*}
        \braket{\phi}{O_{AC}^n}{\phi}=\frac{1}{4^n}\int\cdots\int\frac{1}{d}\Tr(A_{U_1}\cdots A_{U_n})\Tr(B_{U_1}\cdots B_{U_n}\sigma_C)
        dU_1\cdots dU_n.
    \end{align*}
    Since $\abs*{\Tr(B_{U_1}\cdots B_{U_n}\sigma_C)}\leq 1$, we can use \cref{lem:exp-bound} to bound this by $$\frac{1}{4^nd}\sqrt{\int\cdots\int \abs*{\Tr(A_{U_1}\cdots A_{U_{n}})}^2dU_1\cdots dU_{n}}.$$
    Using \cref{lem:variance}, we have that $\int\cdots\int \abs*{\Tr(A_{U_1}\cdots A_{U_{n}})}^2dU_1\cdots dU_{n}=\frac{1}{(d^2-1)^{n-1}}+1$, giving
    \begin{align*}
        \braket{\phi}{O_{AC}^n}{\phi}\leq\frac{1}{4^n d}\sqrt{1+\frac{1}{(d^2-1)^{n-1}}}\leq\frac{\sqrt{2}}{4^nd}.
    \end{align*}
\end{proof}

\begin{corollary}\label{cor:all-bc}
    Let $O_{AC}$ and $\ket{\phi}$ be as in the above lemma. Let $O_{BC}=\frac{1}{4}\int I\otimes B_U\otimes B_U dU$, $\Pi=\frac{I+F}{2}$ for $F$ the swap operation between systems $B$ and $C$, and $\tilde{O}=2\Pi O_{AC}+O_{BC}$. Then, for all $n\geq 0$ there exists a vector $\ket{\zeta_{n}}$ such that $\norm{\ket{\zeta_{n}}}\leq 1$ and $\tilde{O}^n\ket{\phi}=O_{BC}^n\ket{\phi}+\frac{4}{\sqrt{d}}\ket{\zeta_{n}}$.
\end{corollary}

\begin{proof}
    Using \cref{lem:near-eigenstate}, we have that $\norm{O_{AC}\ket{\phi}}\leq\frac{1}{2\sqrt{d}}$ since $\norm{O_{AC}\ket{\phi}}^2=\braket{\phi}{O_{AC}^2}{\phi}\leq\frac{\sqrt{2}}{4^2d}\leq\frac{1}{4d}.$ Let $a_n=4-4\parens*{\frac{3}{4}}^n$. We proceed by induction to show $\tilde{O}^n\ket{\phi}=O_{BC}^n\ket{\phi}+\frac{a_n}{\sqrt{d}}\ket{\zeta_{n}'}$ with $\norm{\ket{\zeta'_n}}\leq 1$, which implies the wanted result by taking $\ket{\zeta_n}=\frac{a_n}{4}\ket{\zeta'_n}$. For $n=0$, the wanted result trivially holds by taking $\ket{\zeta_{0}'}=0$. Now, suppose that the result holds for $n-1$, that is $\tilde{O}^{n-1}\ket{\phi}=O_{BC}^{n-1}\ket{\phi}+\frac{a_{n-1}}{\sqrt{d}}\ket{\zeta_{n-1}'}$. Therefore,
    \begin{align*}
        \tilde{O}^{n}\ket{\phi}&=\tilde{O}O_{BC}^{n-1}\ket{\phi}+\frac{a_{n-1}}{\sqrt{d}}\tilde{O}\ket{\zeta_{n-1}'}\\
        &=O_{BC}^{n}\ket{\phi}+2\Pi O_{AC}O_{BC}^{n-1}\ket{\phi}+\frac{a_{n-1}}{\sqrt{d}}\tilde{O}\ket{\zeta_{n-1}'}.
    \end{align*}
    Expanding $O_{AC}O_{BC}^{n-1}\ket{\phi}$, we get
    \begin{align*}
        O_{AC}O_{BC}^{n-1}\ket{\phi}=\int\cdots\int O_{AC}(I\otimes B_{U_1}\cdots B_{U_n}\otimes B_{U_1}\cdots B_{U_n})\ket{\phi}dU_1\cdots dU_n;
    \end{align*}
    The norm of each term in this integral is $\frac{1}{2\sqrt{d}}$, and therefore $$\norm{2\Pi O_{AC}O_{BC}^{n-1}\ket{\phi}}\leq2\norm{O_{AC}O_{BC}^{n-1}\ket{\phi}}\leq \frac{1}{\sqrt{d}}.$$
    Let $\ket{\zeta_n'}=\frac{\sqrt{d}}{a_n}\parens*{2\Pi O_{AC}O_{BC}^{n-1}\ket{\phi}+\frac{a_{n-1}}{\sqrt{d}}\tilde{O}\ket{\zeta_{n-1}'}}$. By construction, $\tilde{O}^n\ket{\phi}=O_{BC}^n\ket{\phi}+\frac{a_n}{\sqrt{d}}\ket{\zeta_n'}$. Finally, since $\norm{\tilde{O}}\leq\frac{3}{4}$  and we have that
    \begin{align*}
        \norm{\ket{\zeta_n'}}&\leq\frac{\sqrt{d}}{a_n}\parens*{\norm{2\Pi O_{AC}O_{BC}^{n-1}\ket{\phi}}+\norm{\frac{a_{n-1}}{\sqrt{d}}\tilde{O}\ket{\zeta_{n-1}'}}}\\
        &\leq\frac{1}{a_n}\parens*{1+\frac{3}{4}a_{n-1}}=\frac{1+\frac{3}{4}\parens*{4-4\parens*{\frac{3}{4}}^{n-1}}}{4-4\parens*{\frac{3}{4}}^{n}}\\
        &=1.
    \end{align*}
\end{proof}

\begin{theorem}\label{thm:inner-product-bound}
    Let $O=\frac{1}{4}\int(UAU^\dag\otimes I\otimes B_U+UAU^\dag\otimes B_U\otimes I+I\otimes B_U\otimes B_U)dU$ be the game polynomial. Define $\varepsilon\in\mathbb{R}$ such that $\|O\|=\frac{1}{4}+\varepsilon$. Suppose $\rho_{ABC}=\ketbra{\psi}_{ABC}$ is the maximum-eigenvalue eigenstate of $O$ with eigenvalue $\frac{1}{4}+\varepsilon$.  Let $\sigma_{ABC}=\ketbra{\phi}_{ABC}$ be such that $\sigma_{AC}=\omega_A\otimes\sigma_C$. Then, assuming $\varepsilon > 0,$
    \begin{align}
        \abs*{\braket{\psi}{\phi}}^2\leq\frac{4}{2^{(2\log d-8)\varepsilon}}
    \end{align}
\end{theorem}

\begin{proof}
    First, note that, in the notation of the previous corollary, $O=O_{AC}+FO_{AC}F+O_{BC}$. We may assume without loss of generality that $\ket{\psi}$ is symmetric under $F$, so
    $$\parens*{\frac{1}{4}+\varepsilon}\ket{\psi}=(O_{AC}+FO_{AC}F+O_{BC})\ket{\psi}=((I+F)O_{AC}+O_{BC})\ket{\psi}=\tilde{O}\ket{\psi}.$$
    Next, fix $n\in\N$. Then, using \cref{lem:near-eigenstate} and the fact that $\norm{O_{BC}}\leq\frac{1}{4}$
    \begin{align*}
        \abs*{\braket{\phi}{\psi}}^2&=\frac{1}{\parens*{\frac{1}{4}+\varepsilon}^{2n}}\abs*{\braket{\phi}{\tilde{O}^n}{\psi}}^2\\
        &=\frac{\abs*{(\bra{\phi}O_{BC}^n+\bra{\zeta_n}\frac{4}{\sqrt{d}})\ket{\psi}}^2}{\parens*{\frac{1}{4}+\varepsilon}^{2n}}\\
        &\leq\frac{\parens*{\frac{1}{4^n}+\frac{4}{\sqrt{d}}}^2}{\parens*{\frac{1}{4}+\varepsilon}^{2n}}=\frac{\parens*{1+\frac{4^{n+1}}{\sqrt{d}}}^2}{\parens*{1+4\varepsilon}^{2n}}.
    \end{align*}
    Set $n=\frac{\log d}{4}-1$, and note that $\parens*{1+4\varepsilon}^{\frac{1}{4\varepsilon}}\geq 2$, so we get
    \begin{align*}
        \abs*{\braket{\phi}{\psi}}^2&\leq\frac{4}{2^{8n\varepsilon}}=\frac{4}{2^{(2\log d-8)\varepsilon}}.
    \end{align*}
\end{proof}

\begin{corollary}\label{cor:entropy-bound}
    Let $\rho$ be as in \cref{thm:inner-product-bound}. Then,
    \begin{align}
        H_{\min}(A|B)_\rho\geq-\parens*{1-2\varepsilon}\log d-8\varepsilon-2
    \end{align}
\end{corollary}

\begin{proof}
    Recall from \cref{lem: op-int-min-ent} that \begin{align*}
        2^{-\log d-H_{\min}(A|B)_\rho}&=\max_{\mc{E}}F\parens*{(\id\otimes\mc{E})(\rho_{AB}),\ketbra{\phi^+}}^2\\
        &=\max_{\mc{E}}\braket{\phi^+}{(\id\otimes\mc{E})(\rho_{AB})}{\phi^+},
    \end{align*}
    where the maximisation is over channels $\mc{E}:B(\mc{H}_B)\rightarrow B(\mc{H}_{A'})$. Purifying, we find $$2^{-\log d-H_{\min}(A|B)_\rho}=\max_{V,\ket{v}}\abs*{(\bra{\phi^+}\otimes\bra{v})(\mds{1}\otimes V\otimes \mds{1})\ket{\psi}}^2,$$ where the maximisation is over isometries $V:\mc{H}_B\rightarrow\mc{H}_{A'E}$ and states $\ket{v}\in\mc{H}_{EC}$. Let $\ket{\phi}=\ket{\phi^+}\otimes\ket{v}$, and note that $(\mds{1}\otimes V\otimes \mds{1})\ket{\psi}$ is an eigenvector of $(\mds{1}\otimes V\otimes \mds{1})O'(\mds{1}\otimes V^\dag\otimes \mds{1})$ with eigenvalue $\geq\frac{1}{2}+\varepsilon$. Hence, using \cref{thm:inner-product-bound}, we find that
    \begin{align*}
        \abs*{(\bra{\phi^+}\otimes\bra{v})(\mds{1}\otimes V\otimes \mds{1})\ket{\psi}}^2&\leq\frac{4}{2^{(2\log d-8)\varepsilon}}.
    \end{align*}
    Then, we have $2^{-\log d-H_{\min}(A|B)_\rho}\leq\frac{4}{2^{(2\log d-8)\varepsilon}}$. Rearranging gives the result.
\end{proof}

In the following theorem, we bound the winning probability of the Haar-measure game.

\begin{theorem}\label{thm:D}
    Let $d> 2^8$ be even. The quantum value of the $d$-dimensional two-outcome Haar measure game $\ttt{G}_{d,2}$ (see \cref{def:haar-moe}) is
    \begin{align}
        \mfk{w}(\ttt{G}_{d,2})\leq\frac{1}{2}+\frac{\log(\log d-4)}{\log d-4}.
    \end{align}
\end{theorem}

\begin{proof}
    Let $\varepsilon=\frac{\log(\log d-4)}{\log d-4}$. Suppose that $\mfk{w}(\ttt{G}_{d,2})>\frac{1}{2}+\varepsilon$. Then, let $\ttt{S}=(B,C,\{B^U_x\},\{C^U_x\},\rho_{ABC})$ be a strategy such that $\mfk{w}(\ttt{G}_{d,2},\ttt{S})\geq\frac{1}{2}+\varepsilon$. We may, without loss of generality, assume that the strategy is symmetric between exchange of Bob and Charlie, and therefore that $\rho_{ABC}=\ketbra{\psi}_{ABC}$ is an eigenstate of the game polynomial $O$ with eigenvalue $\geq\frac{1}{4}+\varepsilon$. Recall that $A=A_1A_2$, where $A_1=\{0,1\}$ and $A_2=[d/2]$. Then, using \cref{thm:decoupling},
    \begin{align*}
        &\mfk{w}(\ttt{G}_{d,2},\ttt{S})=\int\sum_x \Tr\squ*{(U\Pi_x U^\dag\otimes B^U_x\otimes C^U_x)\rho_{ABC}}dU\\
        &\leq\int\sum_x \Tr\squ*{(U\Pi_x U^\dag\otimes B^U_x)\rho_{AB}}dU\\
        &=\int\sum_x \Tr\squ*{(\ketbra{x}\otimes B^U_x)\Tr_{A_2}((U^\dag\otimes \mds{1})\rho_{AB}(U\otimes \mds{1}))}dU\\
        &\leq\int\sum_x \Tr\squ*{(\ketbra{x}\otimes B^U_x)(\omega_{A_1}\otimes\rho_B)}+\norm{\Tr_{A_2}((U^\dag\otimes \mds{1})\rho_{AB}(U\otimes \mds{1}))-\omega_{A_1}\otimes\rho_B}_{\Tr}dU\\
        &\leq\frac{1}{2}+2^{-\frac{1}{2}H_{\min}(A|B)_\rho-\frac{1}{2}\log d}.
    \end{align*}
    Now, using \cref{cor:entropy-bound}, we find
    \begin{align*}
        \mfk{w}(\ttt{G}_{d,2},\ttt{S})\leq\frac{1}{2}+\sqrt{\frac{4}{2^{(2\log d-8)\varepsilon}}}=\frac{1}{2}+\frac{2}{2^{(\log d -4)\varepsilon}}.
    \end{align*}
    Since $\varepsilon=\frac{\log(\log d-4)}{\log d-4}$, this implies that
    \begin{align*}
        \mfk{w}(\ttt{G}_{d,2},\ttt{S})&\leq\frac{1}{2}+\frac{2}{2^{\log(\log d-4)}}=\frac{1}{2}+\frac{2}{\log d-4}\\
        &<\frac{1}{2}+\frac{\log(\log d-4)}{\log d-4},
    \end{align*}
    a contradiction.
\end{proof}

Now we prove our main result.

\begin{proof}[Proof of~\cref{thm:main-theorem}]
    Applying \cref{lem:moe-bound} and then \cref{thm:D} gives the result. 
\end{proof}

\section{Efficient Construction} \label{sec:eff-const}

\begin{definition} \label{def:efficient-qecm}
    Let $A$, $\sigma_0$, and $\sigma_1$ be in 
    \cref{def:haar-qecm}, and let $\mc{V} \subseteq \mc{U}(\mc{H}_A)$ be a finite set. The \emph{encryption of a bit induced by $\mc{V}$} is the QECM $\ttt{Q}_{\mc{V},2}=(\mc{V},\{0,1\},A,\mu_{\mc{V}},\{U\sigma_xU^\dag\}_{U\in \mc{V},x\in\{0,1\}})$, where $\mu_\mc{V}$ is the uniform distribution on $\mc{V}$.
\end{definition}

For $n\in\N$, let $\mc{V}_n\subseteq \mc{U}(2^n)$ be a $2$-design that can be efficiently implemented, which exists by~\cite{DCEL09,CLLW16}. We denote the QECM where the unitaries are sampled from $\mathcal{V}_n$ rather than the full Haar distribution by $\texttt{Q}_{\mathcal{V}_n,2}$; this is formally defined in \cref{def:efficient-qecm}. Then, $\{\ttt{Q}_{\mc{V}_\lambda,2}\}_\lambda$ is an efficient QECM, which is used in the following result.

\begin{theorem} \label{thm: eff-const}
    There exists an efficient QECM encoding a bit that is correct and $\widetilde{O}\parens*{\frac{1}{\lambda}}$-uncloneable secure.
\end{theorem}

\begin{proof}[Proof sketch]
    Let $\mc{V}_n$ be as above. Then $\{\ttt{Q}_{\mc{V}_\lambda,2}\}_\lambda$ is correct and efficient. Further, note that all the arguments leading to the proof of \cref{thm:main-theorem} rely only on the order-$2$ moments of Haar measure, so they hold for any $2$-design. As such, we get that $\ttt{Q}_{\mc{V}_\lambda,2}$ is $\frac{\log(\lambda-4)}{\lambda-4}=\widetilde{O}\parens*{\frac{1}{\lambda}}$-uncloneable secure.
\end{proof}

\section{Outlook} \label{sec:outlook}

In this work, we have shown that uncloneable cryptography is possible without any computational assumptions. We studied the Haar-measure encryption of a bit, which is the QECM where one bit of classical information is encoded as one of the two halves of a uniformly random basis. We showed that this has uncloneable, and hence uncloneable-indistinguishable, security with parameter $O\parens[\big]{\tfrac{\log(\log d)}{\log d}}$. This constitutes a major advancement over previous work, since uncloneable security with parameter tending to $0$ was only known in the quantum random oracle model~\cite{BL20,AKL23}, and unconditional uncloneable security was only known with constant parameter $\approx0.098$~\cite{BBCNPR24}. Further, we show that this security can be attained with an efficient construction by means of $t$-designs.

Our main innovation is in a novel use of the decoupling theorem. In that, it guarantees that a small enough randomly chosen subsystem of a quantum system becomes uncorrelated with an adversarial environment. Note that one might ask how decoupling is applicable in the scenario of uncloneable encryption such that the security we prove is guaranteed. This is because decoupling works by using Haar random unitaries, and in an uncloneable encryption scheme, Alice samples from a Haar distribution. Consider the prepare-and-measure picture where Alice sends mixed states corresponding to messages encoded in a QECM. There may be a reference system $R$, up to an isometry, with which her state is highly entangled resulting in a purification $\xi_{RA}.$ To achieve uncloneable encryption, Alice then basically sends the purification of $R$ to Bob through a noisy channel whose Stinespring dilation is the cloning isometry $V_{A \to BC}$ (without loss of generality, the scheme is symmetric under interchange of Bob and Charlie, so the cloning isometry can also be seen as the dilation of a noisy channel to Charlie). From our application of the decoupling theorem, in the resulting tripartite state $\zeta_{RBC}$, the marginal state of $RC$ decouples as~$\rho_{RC} = \rho_R \otimes \rho_C.$ Thus, $B$ decomposes into subsystems as $B = B_1 B_2$ such that $$\zeta_{RBC} = Z_B (\varphi_{R B_1} \otimes \vartheta_{B_2 C})$$ where $Z_B$ is some unitary change of basis in $B$. This holds most generally as all purifications are isometrically equivalent. Now, to decrypt Bob constructs an isometric decoder $W_{B_1 \to \widetilde{B}} Z_B^{\dag},$ which extracts the purification of $R$ into Bob's preferred subsystem $\widetilde{B}.$ Again, by isometric equivalence of purifications, Bob can choose his decoder to output $\xi_{R\widetilde{B}}$, as a result of which the input state of $RA$ is successfully transmitted to $R\widetilde{B}$ as desired. Therefore, Bob receives full information about~$A$, so only he can recover the encoded message.

We use the properties of the monogamy-of-entanglement game associated with the Haar-measure encryption to guarantee that any state that succeeds with high probability cannot be close to maximally-entangled between the referee and either of the players, whence we can apply decoupling to show that this player becomes completely uncorrelated, and therefore cannot win better than random guessing. The role of decoupling in the proof of information-theoretic security in uncloneable cryptography is telling of its impact, more generally, in quantum cryptography. In hindsight, the fact that the states used to encode messages must be highly mixed in uncloneable encryption is indicative of its inherent connection to a decoupling inequality.

\paragraph{Future directions} We show unconditional uncloneable security with a parameter that scales inverse-logarithmically in the dimension, and hence inverse-polynomially in the security parameter. To achieve the full strength of uncloneable cryptography, this should be improved to a negligible scaling in the security parameter. We believe that this tighter upper bound can be attained, and leave it as an open question for future work.

\pagebreak

\bibliographystyle{bibtex/bst/alphaarxiv.bst}
\bibliography{bibtex/bib/full.bib,bibtex/bib/quantum.bib,bibtex/quantum_new.bib}

\end{document}